\documentclass[11pt]{article}
\usepackage{authblk,algorithm,amsmath,amssymb,amsthm,caption,color,epsfig,latexsym,listings,geometry,graphics,graphicx,psfrag,url}
\usepackage[utf8]{inputenc}
\usepackage[noend]{algpseudocode}
\newtheorem{theorem}{Theorem}[section]
\newtheorem{corollary}{Corollary}[section]
\newtheorem{proposition}{Proposition}[section]
\newtheorem{definition}{Definition}[section]
\newtheorem{lemma}{Lemma}[section]

\newtheorem{example}{Example}[section]

\providecommand{\keywords}[1]{{\tiny\bf Keywords: }{\tiny #1}}

\usepackage{xpatch}
\xapptobibmacro{finentry}{\par\printfield{NOTE}}{}{}

\usepackage{multicol}
\usepackage{stmaryrd} %for \llbracket and \rrbracket
\newcommand{\CC}{\mathbb{C}}

\newcommand{\IPL}{\mathrm{IPL}}

\newcommand{\RR}{\mathbb{R}}

\usepackage{tikz}
\usepackage{stmaryrd}
\usepackage{wrapfig}
\usepackage{cite}

\newcommand{\indicator}[1]{\left\llbracket #1 \right\rrbracket}
\newcommand*\pFqskip{8mu}
\catcode`,\active
\newcommand*\pFq{\begingroup
        \catcode`\,\active
        \def ,{\mskip\pFqskip\relax}
        \dopFq
}
\catcode`\,12
\def\dopFq#1#2#3#4#5{
        {}_{#1}F_{#2}\biggl[\genfrac..{0pt}{}{#3}{#4};#5\biggr]%
        \endgroup
}
\usepackage[colorinlistoftodos,textsize=tiny]{todonotes}
\newcommand{\Comments}{1}
\newcommand{\mynote}[2]{\ifnum\Comments=1\textcolor{#1}{#2}\fi}
\newcommand{\mytodo}[2]{\ifnum\Comments=1%
  \todo[linecolor=#1!80!black,backgroundcolor=#1,bordercolor=#1!80!black]{#2}\fi}

\usepackage{caption}
\usepackage{subcaption}

\title{Sparsification of Phylogenetic Covariance Matrices\\ of $k$-Regular Trees}

\author[1]{Sean Svihla}
\author[1,*]{Manuel E. Lladser}

\affil[1]{Department of Applied Mathematics, University of Colorado, Boulder}
\affil[*]{Corresponding author}

\date{}

\begin{document} 

\maketitle

\begin{abstract}
Consider a tree $T=(V,E)$ with root $\circ$ and edge length function $\ell:E\to\mathbb{R}_+$. The phylogenetic covariance matrix of $T$ is the matrix $C$ with rows and columns indexed by $L$, the leaf set of $T$, with entries $C(i,j):=\sum_{e\in[i\wedge j,o]}\ell(e)$, for each $i,j\in L$. Recent work~\cite{GorLla23} has shown that the phylogenetic covariance matrix of a large, random binary tree $T$ is significantly sparsified with overwhelmingly high probability under a change-of-basis with respect to the so-called Haar-like wavelets of $T$. This finding notably enables manipulating the spectrum of covariance matrices of large binary trees without the necessity to store them in computer memory but instead performing two post-order traversals of the tree. Building on the methods of~\cite{GorLla23}, this manuscript further advances their sparsification result to encompass the broader class of $k$-regular trees, for any given $k\ge2$. This extension is achieved by refining existing asymptotic formulas for the mean and variance of the internal path length of random $k$-regular trees, utilizing hypergeometric function properties and identities.
\end{abstract}

\keywords{cophenetic matrix, Haar-like wavelets, hierarchical data, hypergeometric functions, Lambert W function, metagenomics, phylogenetic covariance matrix, sparsification, ultrametric matrix} 

\section{Introduction}
\label{sec:introduction}

    Hierarchical datasets are described, or presumed to be, by a rooted tree that recursively organizes data into clusters so that its leaves are in a one-to-one correspondence with the data points. Such datasets are common in various fields such as microbial ecology~\cite{cavsfo67}, where they arise as models of inter-species covariance~\cite{pav04,FukMcMDetEtAl12}. In this context, the associated covariance matrices are often large and dense, making their manipulation computationally challenging. Nevertheless, these matrices bear redundancies induced by their hierarchical structure, which may be exploited to sparsify them and make such manipulations tenable, if not trivial.

    Ultrametric matrices, which arise in probability theory and statistical physics among other fields, are often also dense. A symmetric matrix $C\in\mathbb{R}_+^{n\times n}$ is called \textit{ultrametric} when $C(i,j)\ge\min\{C(i,k),C(k,j)\}$, for all $i,j,k\in\{1,\ldots,n\}$. If, in addition, $C(i,i)>C(i,j)$ for all $j\ne i$ when $n>1$, or $C(1,1)>0$ when $n=1$, $C$ is called \textit{strictly ultrametric}. These matrices are fully dense, that is, all their entries are nonzero, but have a myriad of mathematical properties~\cite{MarMicSan94,DelMarSan14}. 

    A matrix $C$ is ultrametric if and only if there is a rooted binary tree $T=(V,E)$ and edge length function $\ell:E\to\RR_+$ such that~\cite{NabVar94,GorLla23}:
    \begin{equation}\label{eqn:cov-mat}
        C=\sum_{e\in E}\ell(e)\,\delta_e\,\delta_e',
    \end{equation}
    where $\delta_e$ is the binary column vector, with entries indexed by the leaves of $T$, indicating the leaves that descend from $e$, and $\delta_e'$ is the transpose of $\delta_e$. If $L$ denotes the leaf set of $T$, the above identity is equivalent to having
    \begin{equation}\label{eqn:cov-mat-ij}
        C(i,j)=\sum_{e\in[i\wedge j,\circ]}\ell(e),\text{ for all }i,j\in L,
    \end{equation}
    where $[i\wedge j,\circ]$ denotes the set of edges that connect $(i\wedge j)$, the least common ancestor of $i$ and $j$, with the root of the tree, denoted from now on as $\circ$. Ultrametric matrices have, therefore, a recursive structure, and their entries are redundant, suggesting that they may be amenable to some form of compression~\cite{GorLla23}. 
    
    The formulation in (\ref{eqn:cov-mat-ij}) arises naturally as a model of phylogenetic covariance wherein the genetic drift of a particular trait follows a Brownian motion~\cite{Har19}. 
    Under this model, each leaf represents a microbial species (or some notion thereof), and the trait variation among different species is a function of time since they diverged evolutionarily. Intuitively, since species sharing more of their evolutionary history should thrive or struggle in similar environments accordingly, a natural measure of trait covariance between two different species $i$ and $j$ is the length of their shared evolutionary history, namely the quantity in (\ref{eqn:cov-mat-ij}).

    In the general setting of rooted trees---not necessarily binary---a matrix with entries such as in (\ref{eqn:cov-mat-ij}) is called the \textit{phylogenetic covariance matrix} (or cophenetic matrix) of a tree. (The term ``phylogenetic'' is usually omitted from now on.) Tree covariance matrices arise naturally in the context of hierarchical datasets; in particular, the class of covariance matrices associated with datasets having a binary hierarchy is precisely the class of ultrametric matrices.

    Recent work~\cite{GorLla23} has demonstrated that, in the case of large datasets with a binary hierarchy, or equivalently, weighted and rooted binary trees, the associated covariance matrices become asymptotically diagonal with overwhelmingly high-probability after changing basis to the so-called Haar-like wavelets~\cite{GavNad10} of the tree. (By ``asymptotically diagonal'', we mean that the fraction of non-zero off-diagonal entries of the covariance matrix, with respect to the wavelet basis, is asymptotically negligible, as the number of tree leaves tends to infinity.) The sparsification of such covariance matrices facilitates manipulations that may be infeasible otherwise. For instance, the spectrum of ultrametric matrices can be derived, whether exactly or approximately depending on the matrix size, from just two post-order traversals of the associated tree without having to store the actual matrix in computer memory~\cite{GorLla23}. In addition, the subclass of ultrametric matrices diagonalized by Haar-like wavelets has been characterized, and their spectra shown to be in bijection with non-negative decreasing functions on the interior nodes of a binary tree~\cite{GorLla23}.

    Nevertheless, many hierarchical datasets, for example, in the context of phylogenomic studies~\cite{ZhuMaiPfeETAL19}, are non-binary and suffer the same downfall of having unmanageably large and dense covariance matrices. It is a natural question, then, whether the same technique may be used to sparsify covariance matrices belonging to a broader class of hierarchy. This extended abstract extends some of the ideas in~\cite{GorLla23} to the broader context of $k$-regular trees, weighted, i.e., rooted trees for which each interior node has exactly $k$ children. Specifically, Theorem~\ref{thm:cov-diag}, Theorem~\ref{thm:kreg-zeta}, Corollary~\ref{cor:zeta->1}, and Corollary~\ref{cor:kreg-sparse} in this manuscript are generalizations of~\cite[Theorem 2.3, Theorem 3.4, Corollary 3.5, and Corollary 3.8]{GorLla23}, respectively, such that they are applicable to any $k\ge2$. In addition, Theorem~\ref{thm:kreg-asym} supplies more precise asymptotic formulas for the mean and variance of the internal path length of random $k$-regular trees---beyond what is currently available in the literature.

\subsection{Notation and Terminology}

    Depending on the context, we regard functions with finite domains as finite-dimensional column vectors, and vice versa. Throughout, $\indicator{\cdot}$ is used to denote indicator functions.
    
    We use standard terminology for trees unless otherwise stated. In particular, $T=(V,E)$ represents a tree rooted at a vertex $\circ\in V$ of degree one and endowed with an edge length function $\ell:E\to\RR_+$. The size of $T$ is the quantity $|T|:=|V|$. The sets of leaves and interior nodes of $T$ are denoted as $L$ and $I$, respectively, and the internal path length of $T$ is the quantity
    \[
        \IPL(T) = \sum_{v\in I}\text{depth}(v).
    \]

    For $u,v\in V$, we denote by $[u,v]$ the set of edges in the shortest path from $u$ to $v$ and by $(u\wedge v)$ the least common ancestor of $u$ and $v$.

    $\mathring{T}$ denotes the \textit{interior} of $T$, obtained by trimming the leaves of $T$.
    
    In our setting, $\circ\notin L$ and $\circ\in I$. (To avoid confusion, since $\circ$ has degree one, it is useful to think of it as an artificial root appended to the root of a tree.) If $v\in V$ has $k$ children, we denote them as $v_1,\ldots,v_k$, and we call a tree \textit{$k$-regular}, $k\ge2$, when its root $\circ$ has degree one and, for all $v\in I\setminus\{\circ\}$, $v$ has $k$ children. 
    
    For a given $v\in V$, $ T(v)$ denotes the sub-tree of $T$ rooted at $v$ and containing all of its descendants. We let $L(v)$, $ I(v)$, and $ E(v)$ denote the sets of leaves, internal nodes, and edges of $ T(v)$, respectively. In addition, we orient edges away from the root, that is, if $e=(u,v)\in E$, then $u$ is understood to be the parent of $v$ (and $v$ is understood to be a child of $u$). We define $T(e):=T(v)$ and $L(e):=L(v)$. 

    The \textit{trace length} of $T$ is the function $\ell^*: E\rightarrow\mathbb{R}_+$ defined as \cite{GorLla23}:
    \[
        \ell^*(e):=|L(e)|\,\ell(e).
    \]
    We define for $u, v\in V$:
    \[
        \ell(u,v) := \sum_{e\in[u,v]}\ell(e),\text{ and }\ell^*(u,v) := \sum_{e\in[u,v]}\ell^*(e),
    \] 
    and, given non-empty $J\subset V$, we denote by $\ell(J, v)$ the column vector with entries $\ell(j,v)$, for all $j\in J$, having dimension $|J|$. We give an analogous definition to $\ell^*(J,v)$.
    
    \textbf{In what follows, unless otherwise stated, $T$ denotes a $k$-regular tree.}

\section{Haar-like Wavelets on $k$-Regular Trees}
\label{sec:haar-like-wavelet-basis}

    In this section, we specialize the Haar-like wavelet basis given in~\cite{GavNad10} to our setting of $k$-regular trees and present a useful interaction between the basis and the phylogenetic covariance matrix associated with any such tree (Theorem~\ref{thm:cov-diag}).

    Wavelets are usually functions defined on a Euclidean space and derive their name from their commonly wave-like shape. They are prevalent in time series and image analysis to localize information across various scales. In our context, the Haar-like wavelets associated with a tree $T$ are functions from $L$ to $\RR$ that are in a one-to-one correspondence with the elements in the set $\{\circ\}\cup\big(I\setminus\{\circ\}\times\{1,\ldots,k-1\}\big)$; in particular, there are $1+(k-1)(|I|-1)$ wavelets associated with a $k$-regular tree. The precise definition follows.

    Given a $v\in I$ and integer $1\le n < k$, define
    \[
        L_{v,n}:=\bigcup_{j=1}^{n}L(v_j),\text{ and }L_{v,n}^+:=L_{v,n}\cup L(v_{n+1}).
    \]

    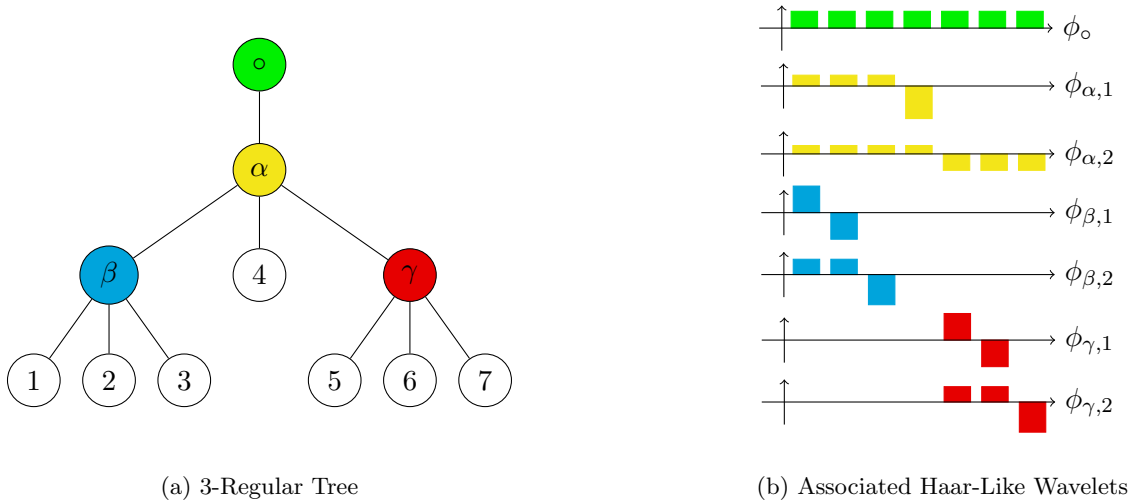
\begin{figure}[t!]
        \captionsetup[subfigure]{font=footnotesize}
        \centering
        \subcaptionbox{3-Regular Tree}[.45\textwidth]{% 
            \begin{tikzpicture}[every node/.style={black, thin, circle}, level distance=1.4cm, level/.style={sibling distance=10mm}, level 2/.style={sibling distance=20mm}]
            \node[draw,fill=green!94!black,minimum size=.7cm] {$\circ$}
            child {node[draw,fill=yellow!94!black,minimum size=.7cm] {$\alpha$}
                child {node[draw,fill=cyan!90!black,minimum size=.7cm] {$\beta$}
                    child {node[draw,minimum size=.7cm] {1}}
                    child {node[draw,minimum size=.7cm] {2}}
                    child {node[draw,minimum size=.7cm] {3}}
                }
                child {node[draw,minimum size=.7cm] {4}}
                child {node[draw,fill=red!90!black,minimum size=.7cm] {$\gamma$}
                    child {node[draw,minimum size=.7cm] {5}}
                    child {node[draw,minimum size=.7cm] {6}}
                    child {node[draw,minimum size=.7cm] {7}} 
                }
            };
            \end{tikzpicture}\vspace{0.67cm}}\hfill%
        \subcaptionbox{Associated Haar-Like Wavelets}[.45\textwidth]{  
            \begin{tikzpicture}
                \draw[->] (-0.3,-0.3) -- (-0.3,0.3) node[left] {};
                \draw[draw, fill=black,color=green!94!black] plot[ybar] coordinates{(0,0.2235) (.5,0.2235) (1,0.2235) (1.5,0.2235) (2,0.2235) (2.5,0.2235) (3,0.2235)};
                \draw[->] (-0.6,0) -- (3.3,0) node[right] {$\phi_\circ\,\,\,\,\,$};
            \end{tikzpicture}
            \begin{tikzpicture}
                \draw[->] (-0.3,-0.3) -- (-0.3,0.3) node[left] {};
                \draw[draw, fill=black,color=yellow!94!black] plot[ybar] coordinates{(0,0.144) (.5,0.144) (1,0.144) (1.5,-0.433)};
                \draw[->] (-0.6,0) -- (3.3,0) node[right] {$\phi_{\alpha,1}$};
            \end{tikzpicture}
            \begin{tikzpicture}
                \draw[->] (-0.3,-0.3) -- (-0.3,0.3) node[left] {};
                \draw[draw, fill=black,color=yellow!94!black] plot[ybar] coordinates{(0,0.1115) (.5,0.1115) (1,0.1115) (1.5,0.1115) (2,-0.2182) (2.5,-0.2182) (3,-0.2182)};
                \draw[->] (-0.6,0) -- (3.3,0) node[right] {$\phi_{\alpha,2}$};
            \end{tikzpicture}
            \begin{tikzpicture}
                \draw[->] (-0.3,-0.3) -- (-0.3,0.3) node[left] {};
                \draw[draw, fill=black,color=cyan!90!black] plot[ybar] coordinates{(0,0.3535) (.5,-0.3535)};
                \draw[->] (-0.6,0) -- (3.3,0) node[right] {$\phi_{\beta,1}$};
            \end{tikzpicture}
            \begin{tikzpicture}
                \draw[->] (-0.3,-0.3) -- (-0.3,0.3) node[left] {};
                \draw[draw, fill=black,color=cyan!90!black] plot[ybar] coordinates{(0,0.204) (.5,0.204) (1,-0.4011)};
                \draw[->] (-0.6,0) -- (3.3,0) node[right] {$\phi_{\beta,2}$};
            \end{tikzpicture}
            \begin{tikzpicture}
                \draw[->] (-0.3,-0.3) -- (-0.3,0.3) node[left] {};
                \draw[draw, fill=black,color=red!90!black] plot[ybar] coordinates{(2,0.3535) (2.5,-0.3535)};
                \draw[->] (-0.6,0) -- (3.3,0) node[right] {$\phi_{\gamma,1}$};
            \end{tikzpicture}
            \begin{tikzpicture}
                \draw[->] (-0.3,-0.3) -- (-0.3,0.3) node[left] {};
                \draw[draw, fill=black,color=red!90!black] plot[ybar] coordinates{(2,0.204) (2.5,0.204) (3,-0.4011)};
                \draw[->] (-0.6,0) -- (3.3,0) node[right] {$\phi_{\gamma,2}$};
            \end{tikzpicture}\vspace{0.33cm}}\\
        \caption{The Haar-like wavelets associated with a 3-regular tree. The wavelets are orthonormal, piece-wise constant functions, with domain $1 \leq n \leq 7$ and, except for $\phi_0$, they are each mean zero.}
        \label{fig:haar-wavelets}
    \end{figure}

    \begin{definition}[Haar-like wavelets of a $k$-regular tree]
    \label{dfn:haar-basis}
        The (mother) wavelet associated with the root of $T$ is the function $\phi_\circ:L\to\RR$ defined
        \[
            \phi_\circ(i) := \frac{1}{\sqrt{|L|}},\text{ for }i\in L.
        \]
        For each $v\in I\setminus\{\circ\}$ and integer $1\le n<k$, the wavelet associated with $(v,k)$ is the function $\phi_{v,n}:L\to\RR$ defined as
        \[
            \phi_{v,n}(i):= \sqrt{\frac{|L(v_{n+1})|}{|L_{v,n}|\cdot|L_{v,n}^+|}}\indicator{i\in L_{v,n}}-\sqrt{\frac{|L_{v,n}|}{|L(v_{n+1})|\cdot|L_{v,n}^+|}}\indicator{i\in L(v_{n+1})}
        \]
        for $i\in L$.
    \end{definition}

    The Haar-like wavelets are mutually orthogonal~\cite{GavNad10}. Moreover, from well-known facts about graphs and trees,
    \[
        |I|+|L|-1=\frac{1}{2}\sum_{v\in V}\text{deg}(v)=\frac{1}{2}\Big[|I|+|L|+k\big(|I|-1\big)\Big],
    \]
    hence $|L|=1+(k-1)(|I|-1)$. That is, there are as many Haar-like wavelets as leaves on the tree, and the wavelets form therefore an orthonormal basis of the linear space of functions from $L$ to $\RR$. Further, $\phi_{v,n}$ assumes one of two distinct values within its support, which are chosen so that $\phi_{v,n}$ has unit $\ell^2$-norm and mean zero:
    \begin{equation*}\label{eqn:wavelet-sum-zero}
        \sum_{i\in L}\phi^2_{v,n}(i) = 1\text{ and }\sum_{i\in L}\phi_{v,n}(i) = 0.
    \end{equation*}
    To help fix ideas, Figure \ref{fig:haar-wavelets} provides an illustration of the Haar-like wavelets.

    Notice that $\phi_{v,n}$ is supported on $L_{v,n}^+$. As such, wavelets associated with nodes nearer the root capture coarser information about the leaves of the tree. As noted in \cite{GorLla24} for $k=2$, given a function $\varphi:L\to\RR$ and wavelet $\phi_{v,n}$, we have
    \begin{equation}\label{ide:innerprodA}
        \langle\varphi,\phi_{v,n}\rangle=c_{v,n}\cdot\left\{\frac{1}{|L_{v,n}|}\sum_{i\in L_{v,n}}\varphi(i)-\frac{1}{|L(v_{n+1})|}\sum_{i\in L(v_{n+1})}\varphi(i)\right\},
    \end{equation}
    where 
    \begin{equation}\label{ide:innerprodB}
        c_{v,n}:=\sqrt{\frac{|L_{v,n}|\,|L(v_{n+1})|}{|L_{v,n}^+|}}.
    \end{equation}
    Consequently, projecting a function onto $\phi_{v,n}$ is the same, up to a constant factor, as computing the difference between the average values of $\varphi$ over $L_{v,n}$ and $L(v_{n+1})$. Hence the coordinates of a real-valued function defined over $L$ with respect to the Haar-like wavelets can be computed efficiently, which is relevant for applications involving large trees.

    The following result highlights a remarkably simple action of the covariance matrix of a $k$-regular tree over its Haar-like wavelets. This property was first noticed in~\cite{GorLla23} for the case of 2-regular trees.
    
    \begin{theorem}
    \label{thm:cov-diag}
        If $\psi$ is a Haar-like wavelet associated with $v\in I$, then $C\,\psi = \mathrm{diag}(\ell^*(L, v))\,\psi$. 
    \end{theorem}

    \begin{proof}
        We first show the result for $\psi=\phi_\circ$. For this, note that for all $j\in L$
        \[
            (C\psi)(j)=\frac{1}{\sqrt{L}}\sum_{i\in L}\ell(i\wedge j,\circ)=\psi(j)\sum_{i\in L}\sum_{e\in[i\wedge j,\circ]}\ell(e),
        \]
        yet we have the logical equivalence:
        \begin{equation}\label{logeq:key}
            \forall i,j\in L\:\forall v\in I, e\in[i\wedge j,v]\Longleftrightarrow i\in L(e)\text{ and }e\in[j,v].
        \end{equation}
        Hence, $\ell(e)$ occurs $|L(e)|$ times in the previous double-sum, and
        \[
            (C\psi)(j)=\psi(j)\sum_{e\in[j,v]}|L(e)|\,\ell(e)=\psi(j)\,\ell^*(j,v),
        \]
        which is precisely the $j$-th entry of the vector $\mathrm{diag}(\ell^*(L, v))\,\psi$.

        Next, we consider $\psi=\phi_{v,n}$, for a $v\in I\setminus\{\circ\}$ and $1\le n<k$. Then, for each $j\in L$, we have the following cases:
        \begin{enumerate}
            \item[(a)] Assume $j\not\in L_{v,n}^+$. Then, for all $i\in L_{v,n}^+$, $(i\wedge j) = (v\wedge j)$, so
            \begin{align*}
                (C\psi)(j) &= \sum_{i\in L_{v,n}^+}\ell(i\wedge j, \circ)\,\phi_{v,n}(i) \\
                &= \ell(v\wedge j,\circ)\,\sum_{i\in L_{v,n}^+}\phi_{v,n}(i) =0,
            \end{align*}
            because $\phi_{v,n}$ has mean zero. On the other hand, the entry associated with $j$ in the vector $\mathrm{diag}(\ell^*(L, v))\,\psi$ also vanishes because $\psi$ is supported on $L_{v,n}^+$.
            \item[(b)] Finally, assume $j\in L_{v,n}^+$. Then 
            \begin{align*}
                (C\psi)(j)
                &= \sum_{i\in L_{v,n}^+}\ell(i\wedge j, \circ)\,\phi_{v,n}(i)\\
                &= \sum_{i\in L_{v,n}^+}\ell(i\wedge j, v)\,\phi_{v,n}(i)  + \ell(v, \circ)\sum_{i\in L_{v,n}^+}\phi_{v,n}(i)\\
                &=\sum_{i\in L_{v,n}^+} \phi_{v,n}(i) \sum_{e\in[i\wedge j,v]}\ell(e),
            \end{align*}
            where for the middle identity we have reused that $\phi_{v,n}$ has mean zero. But notice that $(i\wedge j) = v$ when $\ell(i\wedge j, v) = 0$. Otherwise, if $(i\wedge j)\ne v$ then $\phi_{v,n}(i) = \phi_{v,n}(j)$. Hence, reusing the logic equivalence in (\ref{logeq:key}), we find that
            \[
                (C\psi)(j) =\phi_{v,n}(j) \sum_{i\in L_{v,n}^+:(i\wedge j)\neq v}\quad\sum_{e\in[i\wedge j,v]}\ell(e)=\psi(j)\,\ell^*(j,v).
            \]
            which corresponds to the entry associated with $j$ in the vector $\mathrm{diag}(\ell^*(L, v))\,\psi$.
        \end{enumerate}
        The theorem now follows from parts (a)-(b). \qedhere
    \end{proof}
    
\section{Sparsification of k-Regular Covariance Matrices}
\label{sec:sparsification-of-covariance-matrices}

    We now leverage the insights from the previous section to obtain a lower bound on the proportion of entries in the covariance matrix of a $k$-regular tree that vanish when switching to the Haar-type basis. To establish our main result in this section, we require the following definition, which is analogous to the one used in~\cite{GorLla23} for 2-regular trees. 

    \begin{definition}[Haar-like matrix of a $k$-regular tree]
    \label{dfn:haar-matrix}
        The Haar-like matrix of $T$ is the square matrix $\Phi$ whose columns are the wavelets associated with the tree; that is, its rows are indexed by $L$ and its columns by the wavelets.
    \end{definition}

    Let $\psi_u, \psi_v$ be wavelets associated with interior nodes $u, v\in I$, respectively. We have as a direct consequence of Theorem~\ref{thm:cov-diag} that
    \begin{equation}\label{ide:magic}
        \big(\Phi'C\Phi\big)(\psi_u,\psi_v) = \psi_u\,\mathrm{diag}\big(\ell^*(L,v)\big)\,\psi_v.
    \end{equation}
    The support of $\mathrm{diag}\big(\ell^*(L,v)\big)\,\psi_v$ is contained in the support of $\psi_v$; therefore, the entry associated with row $\psi_u$ and column $\psi_v$ of $\Phi'C\Phi$ vanish when $\psi_u$ and $\psi_v$ have disjoint supports. Since $\Phi'\,C\,\Phi$ corresponds to $C$ after changing basis to the Haar-like wavelets, (\ref{ide:magic}) shows that the wavelets may be used to sparsify the covariance matrix of $k$-regular trees. Of course, there is no reason why such interactions should occur often enough to meaningfully sparsify the matrix. As detailed next, however, a minimum level of sparsification is guaranteed by the tree's size and internal path length. 

    We emphasize that the following result is a conservative bound on the number of vanished entries under the Haar-like basis. In practice, it is not uncommon to see better sparsification; however, we have found in all cases that the overwhelming majority of sparsification arises from the interactions considered in the following theorem. We discuss this topic further in Section \ref{sec:discussion}.

    \begin{theorem}\label{thm:kreg-zeta}
        Let $\zeta$ denotes the fraction of vanishing entries of $\Phi'C\Phi$. Then
        \[
            1 - \zeta\leq (k-1)^2\frac{1}{|I|} + 2(k-1)^2\frac{\IPL(T)}{|I|^2}.
        \]
    \end{theorem}

    \begin{proof}
        Note that $C$ has dimensions $|L|\times|L|$ because $T$ has as many Haar-like wavelets as leaves. Observe if $u,v\in I$ and $\psi_u,\psi_v$ are any of their corresponding wavelets, identity (\ref{ide:magic}) implies
        \[
            \big(\Phi'\,C\,\Phi\big)(\psi_u,\psi_v)=\sum_{i\in L(u)\cap L(v)}\psi_u(i)\,\ell^*(i,v)\,\psi_v(i).
        \]
        As $\big(\Phi'\,C\,\Phi\big)(\psi_u,\psi_v)=0$ when $L(u)\cap L(v)=\emptyset$, then $\zeta$ obeys
        \[
            1-\zeta \leq \frac{(k-1)^2}{|L|^2}\Big|\Big\{(u,v) \in  I\times  I\:;\: L(u)\cap L(v) \neq \emptyset\Big\}\Big|.
        \]
        Notice, however, that $L(u)\cap L(v)\ne\emptyset$ only if $u$ descends from $v$ (or vice versa)~\cite{GorLla23}. Hence, accounting for pairs of the form $(u,u)$ and, when $u\ne v$, $(u,v)$ and $(v,u)$, we obtain
        \begin{align*}
            1-\zeta &\leq\frac{(k-1)^2}{|L|^2}\left(|I| + 2\sum_{u\in  I}\big(|\mathring{T}(u)| - 1\big) \right) \\
            &\leq\frac{(k-1)^2}{|L|^2}\left(2\sum_{u\in I}|\mathring{T}(u)| -|I|\right).
        \end{align*}
        Now, notice $\sum_{u\in I}|\mathring{T}(u)| = \sum_{u\in I}\sum_{v\in I}\indicator{v\in T(u)}$; that is, each node is counted once for each of its ancestors, or         
        \begin{align*}
            \sum_{u\in I}|\mathring{T}(u)| &=\sum_{v\in I}\big(1+\text{depth}(v)\big) \\
            &= |I|+\IPL(T).
        \end{align*}
        Hence we further obtain
        \[
            1-\zeta\leq\frac{(k-1)^2}{|L|^2}\Big(|I|+2\,\IPL(T)\Big).
        \]
        As mentioned in Section \ref{sec:haar-like-wavelet-basis}, $|L|=1+(k-1)(|I|-1)$. Hence $|L|\ge|I|$; otherwise, $|I|<1$, which is not possible because $k\ge2$ and $\circ\in I$. The theorem is now a direct consequence of the above inequality.
    \end{proof}
 
    \begin{corollary}\label{cor:zeta->1}
        If $\IPL( T) \ll |I|^2$ as $| T|\rightarrow\infty$, then $\zeta = 1 - \text{o}(1)$.
    \end{corollary}

    \begin{example}[Perfect $k$-Regular Trees.]
    \label{ex:perfect-kreg}
        These are $k$-regular trees in which every leaf has the same depth. Let $T$ be one such tree with height $(h+1)$. At each depth $j \geq 1$, there are $k^{j-1}$ nodes. So
        \[
            |I| = 1 + \sum_{j=1}^{h}k^{j-1} = \frac{k + k^h - 2}{k - 1},\text{ and } \IPL(T) = \sum_{j=1}^{h}jk^{j-1} = \frac{hk^{h+1} - (h+1)k^h + 1}{(k-1)^2}.
        \] 
         We find that
        \[
            \frac{\IPL(T)}{|I|^2} = \frac{hk^{h+1} - (h+1)k^h + 1}{(k + k^h - 2)^2}\sim\frac{h}{k^{h-1}}.
        \]
        Hence $\IPL(T)=o\big(|I|^2\big)$, as $h\to\infty$, and Corollary \ref{cor:zeta->1} implies that the Haar-like basis asymptotically diagonalizes the covariance matrix of perfect $k$-regular trees as their height tends to infinity. 
    \end{example}
    
    \begin{example}[$k$-Regular Caterpillar Trees]
    \label{ex:caterpillar-kreg}
        These are $k$-regular trees in which the parent of every leaf is a node in a central path graph. Let $T$ be a $k$-regular caterpillar tree of height $(h+1)$ so that its central path has length $h$. Since there is only one node at depth 1, and at each depth $2 \leq j\leq h$, there are $(k-1)$ leaves and one interior node, we have
        \[
            |I| = 2 + h,\text{ and }\IPL(T) = \sum_{j=1}^{h}j = \frac{h(h+1)}{2}.
        \]
        So
        \[
            \frac{\IPL(T)}{|I|^2} = \frac{h(h+1)}{2(h+2)^2}\sim\frac{1}{2},
        \]
        as $h\rightarrow \infty$, and the lower-bound for $\zeta$ in Theorem~\ref{thm:kreg-zeta} is trivial (in fact, strictly negative), and we cannot guarantee that the covariance matrix associated with $T$ is sparsified to a significant degree as $|T|\rightarrow \infty$. 
    \end{example}

    The second example shows that covariance matrices of some $k$-regular trees do not meet the criterion for significant sparsification given by Theorem~\ref{thm:kreg-zeta}; moreover, the tree in Example~\ref{ex:perfect-kreg} is exceptionally constrained. So, the question remains whether trees meeting the sparsification criterion are at all common. To this end, we note that the guarantee provided by Theorem \ref{thm:kreg-zeta} is solely based on the tree's topology; i.e., the edge lengths are irrelevant. Therefore, we may investigate our remaining question by considering random $k$-regular trees irrespective of their edge lengths.
    
\subsection{Interlude on Hypergeometric Functions}

    For a concise introduction to hypergeometric functions, see~\cite{Aba99,WeissteinHypergeometric}.

    A hypergeometric function is one whose power series is hypergeometric; that is, its ratio of consecutive coefficients indexed by $n$ is a rational function of $n$. In particular, a power series $\sum_{n=0}^\infty f_n\,z^n$ is hypergeometric when there are constants $a_1,\ldots,a_p,b_1,\ldots,b_q$, for some integers $p,q\ge1$, such that
    \[
        \frac{f_{n+1}}{f_n} = \frac{1}{n+1}\cdot\frac{(n + a_1)\cdots(n + a_p)}{(n + b_1)\cdots(n + b_q)}.
    \]
  The coefficients $f_n$ may be written in terms of the Pochhammer symbol (defined such that $(c)_0:=1$ and $(c)_n:=\prod_{i=0}^{n-1}(c+i)$ for each integer $n\geq 1$) as
    \[
        f_n=\frac{1}{n!}\cdot\frac{(a_1)_n\cdots(a_p)_n}{(b_1)_n\cdots(b_q)_n},\text{ for all }n\ge0.
    \]
    The above hypergeometric function is denoted ${}_{p}F_{q}(a_1, a_2, \dots, a_p; b_1, b_2, \dots, b_q; z)$. A hypergeometric function of this form is said to be \textit{$s$-balanced} if 
    \[
        \sum_{j=1}^{q}b_j - \sum_{j=1}^{p}a_j = s,
    \]
    and if $s > 0$, the series converges at $z = 1$~\cite{RonSta84}.

    The following result will be crucial to identifying dominant singularities of the generating function enumerating internal path length of $k$-regular trees, as well as addressing their uniqueness in the closure of the disk of convergence. In stating this result, we use the following notation:
    \begin{equation}\label{def:p(t)}
        p(t):=t(1-t)^{k-1},\text{ for all }t\in\CC.
    \end{equation}
    and
    \begin{equation}\label{def:zk}
        z_k:=p\left(k^{-1}\right)=\frac{(k-1)^{k-1}}{k^k}.
    \end{equation}
        
    \begin{proposition}[Reformulation of  {\cite[Identity (25)]{WeissteinHypergeometric}}]
    \label{prop:Glasser}
        If $k\ge2$ then the hypergeometric function ${}_{k-1}F_{k-2}(\frac{1}{k}, \dots, \frac{k-1}{k}; \frac{k-2}{k-1},\frac{k}{k-1}; \frac{z}{z_k})$ is analytic in the disc $|z|<z_k$ and continuous in $|z|\le z_k$. Furthermore,
        \[
            \pFq{k-1}{k-2}{\frac{1}{k}, \dots, \frac{k-1}{k}}{\frac{2}{k-1}, \dots, \frac{k-2}{k-1},\frac{k}{k-1}}{\frac{p(t)}{p\left(\frac{1}{k}\right)}} = \frac{1}{1 - t},\text{ for all }0\leq t\leq\frac{1}{k}.
        \]
    \end{proposition}

    \begin{proof}  
        We attribute some of the ideas in this proof to I. Gessel~\cite{Gessel}.
    
        Consider the functional equation
        \begin{equation}\label{ide:FctEqF(z)}
            F(z)=1+z\{F(z)\}^k,
        \end{equation}
        with $F$ analytic in an open neighborhood of $z=0$. The Lagrange Inversion Theorem implies that  (\ref{ide:FctEqF(z)}) has a unique solution in some open neighborhood of $z=0$, with coefficients given by: 
        \begin{equation}\label{ide:FussCatalan}
            [z^n]\,F(z) = \frac{1}{(k - 1)n + 1}{kn \choose n},
        \end{equation}
        the so-called Fuss-Catalan numbers with parameter $k$. We find
        \[
            \lim_{n\to\infty}\frac{[z^{n+1}]\,F(z)}{[z^n]\,F(z)}=\frac{k^k|z|}{(k-1)^{k-1}},
        \]
        implying that $F(z)$ has radius of convergence $z_k$ and, due to the Vivanti-Pringsheim Theorem~\cite[Theorem 5.7.1]{Hil59}, $z_k$ is a singular point of $F$.
        
        On the other hand, the ratio of consecutive terms in $F(z)$ is
        \begin{align*}
            \frac{z^{n+1}\,[z^{n+1}]\,F(z)}{z^n\,[z^n]\,F(z)}
            &= \frac{\left(n + \frac{1}{k}\right)\cdots\left(n + \frac{k-1}{k}\right)}{\left(n + \frac{2}{k-1}\right)\cdots\left(n + \frac{k-2}{k-1}\right)\cdot\left(n + \frac{k}{k-1}\right)}\cdot\frac{z/z_k}{n + 1},
        \end{align*}
        which reveals
        \begin{equation}\label{ide:hypergeometric}
            F(z) = \pFq{k-1}{k-2}{\frac{1}{k},\dots,\frac{k-1}{k}}{\frac{2}{k-1},\dots,\frac{k-2}{k-1},\frac{k}{k-1}}{\,\frac{z}{z_k}},
        \end{equation}
        being the only solution of (\ref{ide:FctEqF(z)}) in the disk $|z|<z_k$. Moreover, since the balance of this hypergeometric function is
        \[
            s = \left(\sum_{j=2}^{k-2}\frac{j}{k-1}\right) + \frac{k}{k-1} - \sum_{j=1}^{k-1}\frac{j}{k} =\frac{1}{2}>0,
        \]
        it follows from~\cite{RonSta84} that the series of $F(z)$ converges at $z=z_k$. Since $F$ has non-negative coefficients, it follows that the series is absolutely convergent for all $|z|\le z_k$, 
        implying $F$ is analytic for $|z|<z_k$ and continuous for $|z|\le z_k$. 

        To complete the proof, notice that $F(z)>0$ for all $0\le z\le z_k$. From the functional equation (\ref{ide:FctEqF(z)}) we have that
        \[
            z = \left(1 - \frac{1}{F(z)}\right)\left(\frac{1}{F(z)}\right)^{k-1} = p\left(\frac{F(z) - 1}{F(z)}\right),\text{ for all }0\le z\le z_k,
        \]
        but $p'(t)\ge0$ for all $0\le t\le k^{-1}$, with equality only at $t=k^{-1}$, and $p(k^{-1})=z_k$. Therefore,
        \[
            p^{-1}(t)=\frac{F(t)-1}{F(t)},\text{ i.e., }F(t) = \frac{1}{1 - p^{-1}(t)},\text{ for all }t\in\left[0,\frac{1}{k}\right],
        \]
        or, equivalently,
        \[
            F\big(p(t)\big) = \frac{1}{1 - t},\text{ for all }t\in\left[0,\frac{1}{k}\right],\]
        which finalizes the proof. 
    \end{proof} 
    
    Our next result complements the previous one by characterizing part of the range of the polynomial $p(t)$ for $|t|\le k^{-1}$ so long as $k$ is large enough. We recall that $z_k$ is as defined in (\ref{def:zk}).

    \begin{figure}[b!]
        \centering
        \begin{tabular}{ccc}
            \includegraphics[scale=0.2]{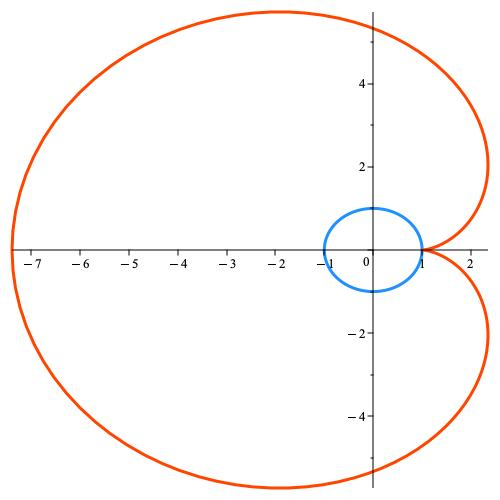} &
            & \includegraphics[scale=0.2]{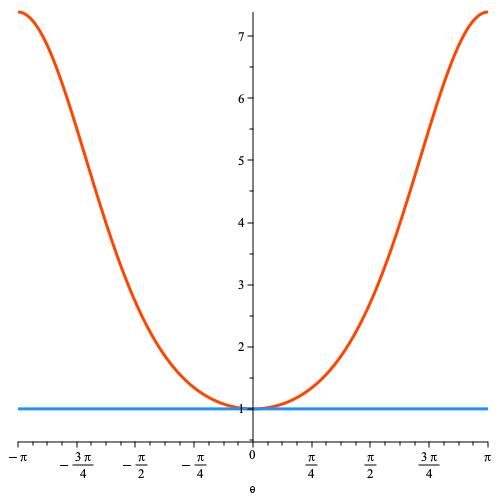}\\  
            (a) & \hspace{1cm} & (b)\\ \\
            \includegraphics[scale=0.2]{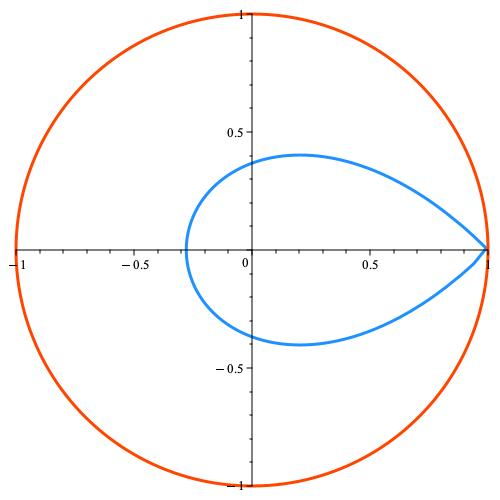}
            & \hspace{1cm} &  
            \includegraphics[scale=0.2]{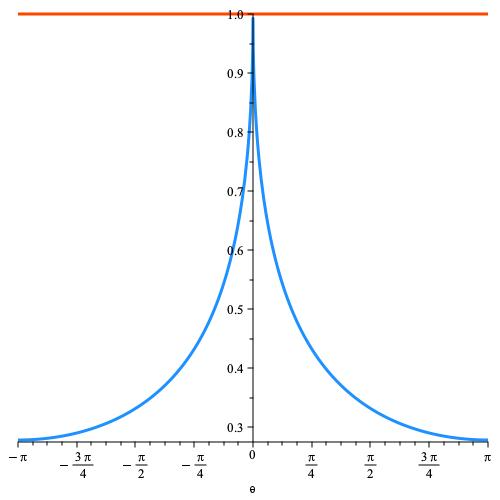}\\
            (c) & \hspace{1cm} & (d)
        \end{tabular}
        \caption{\textbf{(a)} Range of $\tau e^{1-\tau}$ for $|\tau|=1$ (red), and circle $|\tau|=1$ (blue). \textbf{(b)} Plots of $|\tau e^{1-\tau}|$ (red) and $|\tau|$ (blue), with $\tau=e^{i\theta}$, for $-\pi\le\theta\le\pi$. \textbf{(c)} Circle $|\tau|=1$ (red), and range of $-W(-x/e)$ for $|x|=1$ (blue). \textbf{(d)} Plots of $|x|$ (red) and $|W(-x/e)|$ (blue), with $x=e^{i\theta}$, for $-\pi\le\theta\le\pi$ (blue).}
        \label{fig:LambertW}
    \end{figure}
    
    \begin{proposition}
    \label{prop:tricky}
        For all $k$ large enough, if $z\ne z_k$ is such that $|z|=z_k$ then there is a $t$, with $|t|< k^{-1}$, such that $z=p(t)$.
    \end{proposition}
    
    \begin{proof}[Proof]
Consider the auxiliary variables $x=z/z_k$ and $\tau=k\cdot t$. In these variables, the Proposition states that for each $x\ne1$, with $|x|=1$, there is $\tau$, with $|\tau|<1$, such that
        \[
            x=\tau\cdot\left(1+\frac{1-\tau}{k-1}\right)^{k-1}.
        \]
        
        Fix $|x|=1$, $x\ne1$, and consider the functions
        \begin{align*}
            f(\tau)&:=x-\tau e^{1-\tau};\\
            g_k(\tau)&:=\tau e^{1-\tau}-\tau\cdot\left(1+\frac{1-\tau}{k-1}\right)^{k-1}.
        \end{align*}
        
        Since $|\tau e^{1-\tau}|\ge1$ for all $|\tau|=1$, with equality only when $\tau=1$ (see Figure~\ref{fig:LambertW}(a-b)), the continuity of $f$ implies that $\inf_{\tau:\,|\tau|=1}|f(\tau)|>0$. On the other hand, using the exp-log transform, it follows that $\lim_{k\to\infty}g_k(\tau)=0$, uniformly for all $|\tau|=1$. In particular, for all $k$ large enough, we have that
        \[
            \sup_{\tau:|\tau|=1}|g_k(\tau)|<\inf_{\tau:\,|\tau|=1}|f(\tau)|.
        \]
        Roch\'e's Theorem then implies that $f(\tau)+g_k(\tau)$ has the same number of zeroes as $f(\tau)$ in the disk $|\tau|<1$. But the equation $f(\tau)=0$ is equivalent to
        \[
            \frac{x}{e}=\tau\cdot e^{-\tau}.
        \]
        This equation has an infinite number of solutions; however, one of them is $\tau=-W(-x/e)$, where $W$ is the principal branch of the so-called Lambert W function~\cite{COR96}. As seen in Figures~\ref{fig:LambertW}(c-d), however, as long as $x\ne1$, if $|x|=1$ then $|W(-x/e)|<1$, from which the Proposition follows.
    \end{proof}

\subsection{Expectation and Variance of Internal Path Length}

    Corollary~\ref{cor:zeta->1} establishes that a sufficient condition for the Haar-like basis of a $k$-regular tree to diagonalize its (phylogenetic) covariance matrix asymptotically, as the tree grows, is that its internal path length becomes negligible compared to the square of the number of its interior points. To assess the prevalence and extent of trees meeting this criterion, we need asymptotic estimates for the mean and variance of large, uniformly at random, $k$-regular trees, which is precisely what our next result addresses. We emphasize that related results exist in the literature---see~\cite[Proposition VII.3]{FlaSed09}, \cite[Theorem 3.1]{MeiMoo78}, and \cite[Theorem 2.19]{Drm09}. Our new contribution is the derivation of leading asymptotic estimates---with explicit multiplicative constants---in terms of $k$.

    \begin{theorem}\label{thm:kreg-asym}
        For a uniformly at random $k$-regular tree $T$ with $|I|$ internal nodes, the expectation and variance of its internal path length satisfy the asymptotic estimates 
        \begin{equation}\label{thm:EIPLT}
            \mathbb{E}[\IPL(T)] = \sqrt{\frac{\pi k |I|^3}{2(k-1)}}\,\big(1 + \text{O}(|I|^{-1/2})\big),   
        \end{equation}
        and 
        \begin{equation}\label{thm:VIPLT}
            \mathbb{V}[\IPL(T)] = \frac{k}{2(k-1)}\Big(\frac{10}{3} - \pi\Big)|I|^3\,\big(1 + \text{O}(|I|^{-1/2})\big).
        \end{equation}
    \end{theorem}

    \begin{proof}[Proof (Sketch)]
        We establish the formula for the asymptotic expected value and provide only a sketch of the proof for the variance. A detailed proof of the latter will be included in the full version of this extended abstract~\cite{SviLla24+}. Additionally, since the asymptotic formulas in the Theorem agree with those in~\cite{GorLla23} for $k=2$, we assume henceforth that $k\ge3$.
        
        Let $Q(z,u)$ be the generating function for the class of $k$-regular trees, where the variable $z$ marks the number of internal nodes and $u$ the internal path length of a tree. We find using the methods of~\cite[Section III]{FlaSed09} that $Q(z,u)$ obeys the functional equation 
        \begin{equation}\label{ide:Qfcteq}
            Q(z, u) = G\big(z,Q(zu,u)\big)\text{ where } G(z,w):=1+z \,w^k.
        \end{equation}
        For brevity, we define $Q(z):=Q(z,1)$ and $Q_u(z):=\frac{\partial Q}{\partial u}(z,1)$. We are interested in the following:
        \begin{equation}\label{ide:EIPLT}
            \mathbb{E}[\IPL(T)] = \frac{[z^n]Q_u(z)}{[z^n]Q(z)}. 
        \end{equation}
        The asymptotic formula in (\ref{thm:EIPLT}) follows from a detailed asymptotic analysis of the numerator and denominator above. To this effect, we first note from (\ref{ide:Qfcteq}) that
        \begin{equation}\label{ide:FctEqQ(z)}
            Q(z)=1+z\{Q(z)\}^k.    
        \end{equation}
        It follows from the proof of Proposition~\ref{prop:Glasser} that $Q(z)$ is a hypergeometric function of the form
        \begin{equation}
        \label{ide:hypergeometric}
            Q(z) = \pFq{k-1}{k-2}{\frac{1}{k},\dots,\frac{k-1}{k}}{\frac{2}{k-1},\dots,\frac{k-2}{k-1},\frac{k}{k-1}}{\,\frac{z}{z_k}},\text{ for all }|z|\le z_k.
        \end{equation}
        Moreover, since $z_k=p(k^{-1})$, we note the following corollary of Proposition~\ref{prop:Glasser}.

        \begin{corollary}
        \label{cor:Qzk}
            If $k\ge3$ then $Q(z_k) = \frac{k}{k-1}$; in particular, $k\,z_k\,\{Q(z_k)\}^{k-1}=1$.
        \end{corollary}
    
        We next show that $Q(z)$ fits the ``smooth implicit function schema''~\cite{FlaSed09} and provide the asymptotic order of the denominator in (\ref{ide:EIPLT}).
    
        \begin{lemma}
        \label{lemma:schemaQz}
            $z_k$ is the only singularity of $Q(z)$ on the disk $|z|\le z_k$ and, locally around $z_k$, $Q(z)$ admits the representation
            \[
                Q(z)=1+g(z)-h(z)\cdot\sqrt{1-\frac{z}{z_k}},
            \]
            for functions $g(z)$ and $h(z)$ analytic near $z_k$. Furthermore
            \[
                [z^n]\,Q(z) = \sqrt{\frac{k}{2\pi n^3(k-1)^3}}\cdot z_k^{-n}\big(1 + O(n^{-1})\big).
            \]
        \end{lemma}
        
        \begin{proof}
            Define $P(z):=Q(z)-1$. $P(0)=0$ and $P(z)=F\big(z,P(z)\big)$, where $F(z,p):=z\{1+p\}^k$. Note that $F(0,p)=0$ and $F$ has only non-negative Taylor coefficients around $(0,0)$. Moreover, due to Corollary~\ref{cor:Qzk}, $P(z_k)=\frac{1}{k-1}$, hence
            \[
                \begin{array}{clcl}
                F\big(z_k,P(z_k)\big)
                &=\frac{1}{k-1}; & 
                F_z\big(z_k,P(z_k)\big)
                &=\frac{k^k}{(k-1)^k}\ne0;\\
                F_p\big(z_k,P(z_k)\big)
                &= 1 ; &
                F_{pp}\big(z_k,P(z_k)\big)
                &= \frac{(k-1)^2}{k}\ne0.
                \end{array}
            \]
            In particular, $F\big(z_k,P(z_k)\big)=P(z_k)$ and, since $[z^n]\,P(z)>0$ for all $n\ge1$, \cite[Theorem 2.19]{Drm09} implies that that $z_k$ is the only singularity of $P(z)$ and hence of $Q(z)$ in the disk $|z|\le z_k$. Furthermore, it admits the singular expansion
            \[
                P(z)=g(z)-h(z)\cdot\sqrt{1-\frac{z}{z_k}}
            \] 
            where $g(z)$ and $h(z)$ are analytic in an open neighborhood of $z_k$, and
            \[
                [z^n]\,P(z)=\sqrt{\frac{z_k\,F_z\big(z_k,P(z_k)\big)}{2\pi\,F_{pp}\big(z_k,P(z_k)\big)}}\, z_k^{-n}\,n^{-3/2}\,\big(1+\text{O}(n^{-1})\big),
            \]
            from which the lemma follows.
        \end{proof}
        
        We now address the asymptotic behavior of the numerator in (\ref{ide:EIPLT}). For this, note that by implicit differentiation in (\ref{ide:FctEqQ(z)}), we find that
        \begin{equation}\label{ide:Quz}
            Q_u(z) = \frac{k\,z^2\,\{Q(z)\}^{2k - 1}}{\big(1 - kz\,\{Q(z)\}^{k - 1}\big)^2}.
        \end{equation}
    
        Our next result shows that $z_k$ is the only singularity of $Q_u(z)$ in the disk $|z|\le z_k$.
    
        \begin{lemma}
        \label{lem:den}
            The equation $kz\,\{Q(z)\}^{k-1}=1$, with $|z|\le z_k$, has only $z_k$ as a solution.
        \end{lemma}
        
        \begin{proof}
            Per Corollary~\ref{cor:Qzk}, we know that $z_k$ is a solution of the equation $kz\,\{Q(z)\}^{k-1}=1$. On the other hand, because $Q(z)$ has non-negative coefficients, $|kz\,\{Q(z)\}^{k - 1}|<k\,z_k\,\{Q(z_k)\}^{k-1}=1$ for $|z|<z_k$; in particular, any solution to the equation must lie on the circle $|z|=z_k$. But, if $|z|=z_k$ and $z\ne z_k$ then $|Q(z)|<Q(|z|)=Q(z_k)$ because $Q$ has all powers of $z$ with strictly positive coefficients. Hence, $|kz\{Q(z)\}^{k-1}|=kz_k|Q(z)|^{k-1}<kz_k\{Q(z_k)\}^{k-1}=1$, which completes the proof of the lemma.
        \end{proof}
    
        Finally, using (\ref{ide:Quz}) and Lemma~\ref{lem:den}, \cite[Theorem VI.4]{FlaSed09} implies that
        \[
            [z^n]\,Q_u(z, 1) = \frac{k}{2(k-1)^2}\cdot z_k^{-n}\big(1 + \text{O}(n^{-1/2})\big),
        \]
        from which the asymptotic formula for the expected internal path length of a uniformly at random $k$-regular tree in (\ref{thm:EIPLT}) follows.
        
        The asymptotic formula for the variance in (\ref{thm:VIPLT}) follows from a similar analysis based of the identity
        \begin{align*}
            Q_{uu}(z,1) 
            &= \frac{k(7k-1)z^3\{Q(z)\}^{3k-2}}{(1 - kz\{Q(z)\}^{k-1})^{3}}\\
            &\qquad\qquad\qquad+\,\frac{k^2(7k-1)z^4\{Q(z)\}^{4k-3}}{(1 - kz\{Q(z)\}^{k-1})^{4}}
            \,+\,\frac{5k^3(k-1)z^5\{Q(z)\}^{5k-4}}{(1 - kz\{Q(z)\}^{k-1})^5},
        \end{align*}
        also obtained via implicit differentiation in (\ref{ide:FctEqQ(z)}). This completes the proof of Theorem \ref{thm:kreg-asym}.
    \end{proof}
    
%%%%%%%%%%%%%%%%%%%%%%%%%%%%%%%%%%%%%%
%%%%%%%%%%%%%%%%%%%%%%%%%%%%%%%%%%%%%%
%FOR THESIS/FINAL VERSION   
%%%%%%%%%%%%%%%%%%%%%%%%%%%%%%%%%%%%%%
%%%%%%%%%%%%%%%%%%%%%%%%%%%%%%%%%%%%%%
%        
%        It turns out that a general singular expansion can be found by exploiting the similar form of these terms and the singular expansion given by~\cite{Drm09}. We apply ~\cite[Theorem VI.4]{FlaSed09} to obtain
%        \[
%            [z^n]\,Q_{uu}(z,1) = \frac{20k^3(k - 1)z_k^5}{3\sqrt{2^5\pi}}\left(\frac{k}{k-1}\right)^{5k-3/2}\cdot z_k^{-n}n^{3/2}(1 + \text{O}(n^{-1/2})).
%        \]
%        The estimated sequences are sufficient to calculate the expectation and variance of internal path length.

    \begin{figure}[t!]
        \captionsetup[subfigure]{font=footnotesize}
        \centering
        \begin{subfigure}[b]{.48\textwidth}
            \includegraphics[width=1\textwidth]{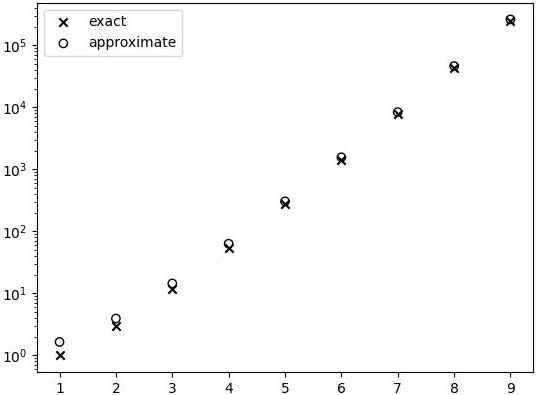}
            \caption{$k=3$}
        \end{subfigure}
        \hfill%
        \begin{subfigure}[b]{.48\textwidth}
            \includegraphics[width=\textwidth]{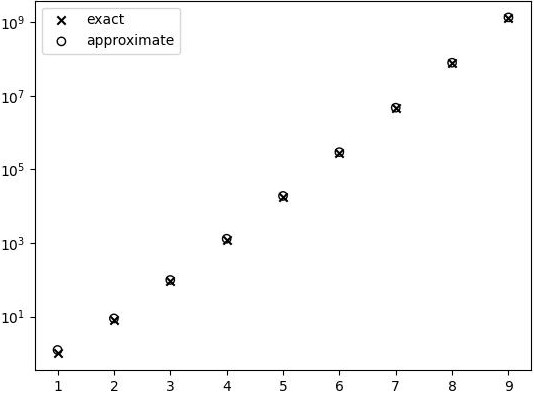}
            \caption{$k=8$}
        \end{subfigure}
        \caption{The first ten coefficients of the counting sequence associated with $Q(z)$, computed both by the exact formula or according to the asymptotic estimate.}
        \label{fig:counting-sequences}
    \end{figure}
    
\subsection{Sparsification of Large, Random, $k$-Regular Covariance Matrices}

    So far, we have seen that the Haar-like wavelets can partially sparsify the (phylogenetic) covariance matrix of a $k$-regular tree by changing the basis, obtaining a lower bound on the proportion of vanishing entries under the basis change. A remaining challenge is determining whether trees meeting the criteria for a high degree of sparsification are common. In this section, we find that such trees are, in fact, abundant and that large random $k$-regular trees are highly sparsified by the Haar-like wavelets with overwhelmingly high probability. 

    Our next result extends~\cite[Corollary 3.8]{GorLla23} to any $k\ge2$. We emphasize that the proof mirrors the one in~\cite{GorLla23}.

    \begin{corollary}
    \label{cor:kreg-sparse}
        Let $k\ge2$ and $T$ be a uniformly at random $k$-regular tree with $|I|$ internal nodes, and $C$ its covariance matrix. If $\Phi$ is the Haar-like matrix associate with $T$, and $\zeta$ represents the fraction of vanishing entries in $\Phi'C\Phi$, then $\zeta\rightarrow1$ in probability, as $|I|\rightarrow\infty$.
    \end{corollary}

    \begin{proof}
        Let $\mu$ and $\sigma^2$ denote the expectation and variance of internal path length of $T$, respectively. Cantelli's inequality states
        \[
            P\big(\IPL( T) > \mu + t\sigma\big) \leq \frac{1}{1 + t^2},\text{ for all }t>0.
        \]
        
        From Theorem~\ref{thm:kreg-asym}, we know that $\mu + t\sigma = \Omega(t|I|^{3/2})$; that is, there is a constant $c > 0$ such that
        \[
            P\big(\IPL(T) > c\,t|I|^{3/2}\big)\leq \frac{1}{1 + t^2},
        \]
        or equivalently
        \[
            P\left(\frac{\IPL(T)}{|I|^2} \leq \frac{c\,t}{\sqrt{|I|}}\right) \geq \frac{t^2}{1 + t^2},\text{ for all }t>0.
        \]
        The result follows by choosing $t\rightarrow\infty$ so that $t = o(\sqrt{|I|})$.
    \end{proof}

    \begin{figure}[t!]
        \captionsetup[subfigure]{font=footnotesize}
        \centering
        \begin{subfigure}[b]{.48\textwidth}
            \centering
            \includegraphics[width=\textwidth]{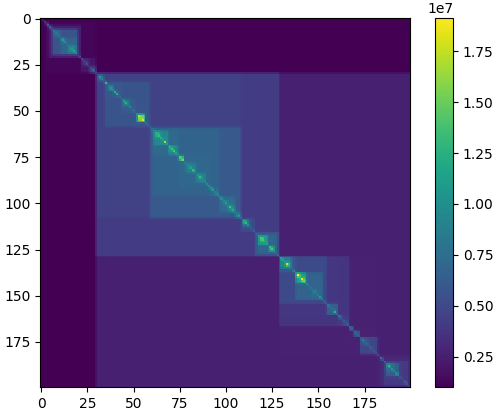}
            \caption{Dense covariance matrix.}
        \end{subfigure}\hfill
        \begin{subfigure}[b]{.48\textwidth}
            \centering
            \includegraphics[width=0.925\textwidth]{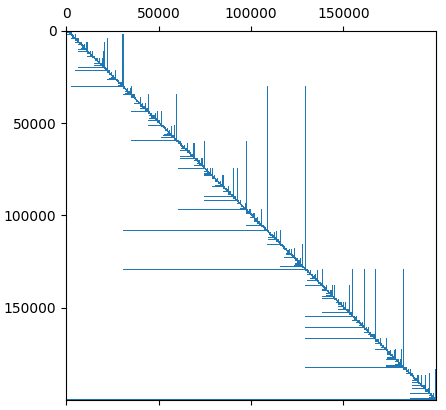}
            \caption{Sparsified covariance matrix.}
        \end{subfigure}
        \caption{(a) Heat-map visualization of a dense covariance matrix $C$, and (b) sparsity pattern of $\Phi'C\Phi$ for a 3-regular tree with 200,001 leaves. The dense matrix has over 40 billion non-zero entries. The sparse matrix, on the other hand, has only about 0.03\% as many non-zero entries. The heatmap of the dense matrix was produced downsampling the sparse representation by a factor of 1000, prior to undoing the change-of-basis.}
        \label{fig:sparsity-pattern}
    \end{figure}

\section{Discussion}
\label{sec:discussion}

    Phylogenetic covariance matrices are often large and dense to the point of being computationally unmanageable; however, we have demonstrated that expressing them in the Haar-like basis can significantly sparsify them. In particular, we have shown that the covariance matrix of a random $k$-regular tree will be highly sparsified with overwhelmingly high probability as the size of the tree tends to infinity. In this section, we illustrate the process of sparsifying a large phylogenetic covariance matrix and discuss some of the practical considerations of the method.

    \begin{wrapfigure}{r}{0.5\textwidth}
        \centering
        \includegraphics[width=0.5\textwidth]{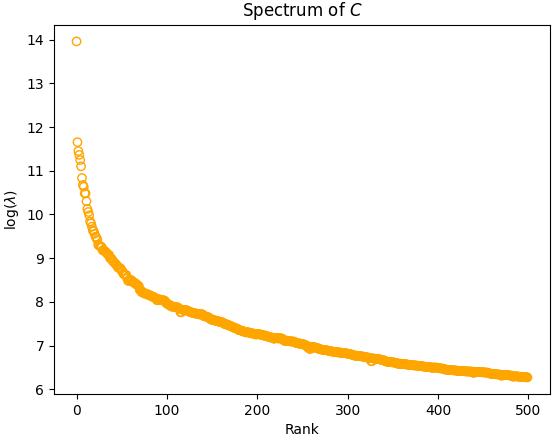}
        \caption{The largest 500 eigenvalues of a covariance matrix $C$ associated with a 3-regular tree with 200,001 leaves. The spectrum was computed from its sparse representation, obtained after changing to the Haar-like basis.}
        \label{fig:spectrum}
    \end{wrapfigure}

    Consider Figure \ref{fig:sparsity-pattern}, which illustrates the application of the Haar-like wavelet transform on a random $3$-regular tree with 200,001 leaves. The tree was generated by simulating a Galton-Watson process until a desired population size (i.e., number of leaves) was reached. We find that, although the original, dense matrix contains over 40 billion non-zero entries, the resulting sparse matrix contains a comparatively meager 11.3 million; that is, over 99.97\% of the entries in the original matrix were zeroed by the transform. Importantly, because $\Phi'C\Phi$ and $C$ are similar matrices, we can quickly calculate the spectrum of $C$ from its sparse representation, illustrated in Figure \ref{fig:spectrum}. 
    
    While this random tree model may not precisely induce actual phylogenetic trees~\cite{Aldous96,AldPit23,aldous2023critical}, we emphasize that the sparsification observed results in an overwhelming part from the tree's topology, irrespective of its edge lengths (which effectively store covariances). Hence, the simulated data gives reasonable insight into the performance of this technique in general.
    
    As already mentioned, the degree of sparsification guaranteed by Theorem~\ref{thm:kreg-zeta} is conservative. For one, it does not account for the fact that the Haar-like wavelets are orthonormal, which would suggest that at least some entries associated with $\psi_u$ and $\psi_v$, where $u$ is an ancestor of $v$, vanish. In fact, if a $k$-regular tree is \textit{trace balanced}~\cite{GorLla23} (that is, for all $v\in I$  and $i,j\in L(v)$, $\ell^*(i, v) = \ell^*(j, v)$), Theorem~\ref{thm:cov-diag} implies that the Haar-like wavelets fully diagonalize its covariance matrix. As a result, we often see a better degree of sparsification than what is guaranteed by Theorem \ref{thm:kreg-zeta}.
         
    % To explore this, we introduce the concept of trace balancedness. Following the definition in~\cite{GorLla23}, we say that a node $v$ is \textit{trace balanced} if 
    % \[
    %     \forall i, j\in L(v), \: \ell^*(i, v) = \ell^*(j, v).
    % \]
    % Since the Haar-like wavelets owe their orthogonality to being mean zero and nested within regions of constant support, then if a node $v$ is trace balanced and $\phi_v$ is a wavelet associated with $v$, then the entries $\Phi'C\Phi(\psi_u, \psi_v)$, for all $u\in I$, are equal to the standard inner product on $\ell^2$ scaled by a constant $\ell^*(i, v)$. Hence, $\Phi'C\Phi(\psi_v, \psi_v)$ is the only non-zero entry in its column. In fact, for $k\geq 3$, we do not even need for $v$ to be trace balanced; rather, it is sufficient that $v$ be "partially trace balanced" on the support of $\psi_v$. We note, as mentioned in~\cite{GorLla23}, that if $v$ is trace balanced for all $v\in T$, then $\Phi'C\Phi$ is diagonal.

    % In addition, the guarantee does not account for the fact that $\psi_v$ does not always have its support equal to $L(v)$, yet $\psi_v$ is considered to have intersecting support with all wavelets $\psi_u$ where $u$ descends from $v$. In fact, a wavelet $\phi_{v,i}$, $i \leq k-2$ has disjoint support with each wavelet in $T(v_{j})$, $j \geq i+2$. Each of these interactions represents a vanishing entry, and as a whole, particularly for nodes with many descendants, these can contribute a non-negligible amount of sparsification. 

    Our analysis shows that the technique of sparsifying dense phylogenetic covariance matrices by a change-of-basis with the Haar-like wavelets extends to the broader class of $k$-regular trees. While further work is required to verify the exact performance of this method on random $k$-ary trees (i.e., ones for which each interior node contains \textit{at most} $k$ children), well-known properties of generating functions enumerating simple varieties of trees suggest that such a generalization is possible by the methods employed here. Comparison with $k$-regular trees indicates sparsification for what one may call “almost” $k$-regular trees (i.e., $k$-ary trees which are $k$-regular except at a relatively small number of internal nodes); however, initial investigation suggests that the worst-case $k$-ary tree might result in poor sparsification. Such research would give access to new phylogenetic datasets with more complicated but richer hierarchical structures.
    
\section*{Acknowledgements}
{This research has been partially funded by the NSF grant No. 1836914.}

%%
%% Bibliography
%%

%% Please use bibtex, 
%\bibliographystyle{abbrv}
%\bibliography{bib2doi}

\end{document}